\documentclass{amsart}[11pt]
\usepackage{amsmath, amsfonts, amsthm, amssymb,mathtools}
\usepackage{subfigure}
\usepackage{hyperref}
\usepackage{color}
\usepackage{graphicx}
\usepackage{enumerate}
\usepackage[ruled, vlined]{algorithm2e}

\numberwithin{equation}{section}
\newtheorem{theorem}{Theorem}[section]

\newtheorem{lemma}[theorem]{Lemma}
\newtheorem{proposition}[theorem]{Proposition}

\theoremstyle{definition}
\newtheorem{definition}[theorem]{Definition}
\newtheorem{remark}[theorem]{Remark}

\newcommand{\diff}[1]{\hspace{0.2em} \D #1}
\newcommand{\ind}{1\hspace{-2.1mm}{1}} 

\newcommand{\Cc}{\mathcal{C}}
\newcommand{\Kk}{\mathcal{K}}
\newcommand{\Ll}{\mathcal{L}}
\newcommand{\Mm}{\mathcal{M}}
\newcommand{\Pp}{\mathcal{P}}
\newcommand{\Tt}{\mathcal{T}}
\newcommand{\Vv}{\mathcal{V}}
\newcommand{\Xx}{\mathcal{X}}

\newcommand{\mm}{\mathrm{m}}
\newcommand{\MMD}{\mathfrak{D}}

\newcommand{\ii}{\boldsymbol{\mathrm{i}}}
\newcommand{\vv}{\mathrm{v}}
\newcommand{\ww}{\mathrm{w}}
\newcommand{\cc}{\mathrm{c}}
\newcommand{\xx}{\mathrm{x}}
\newcommand{\yy}{\mathrm{y}}

\newcommand{\Xf}{\mathfrak{X}}

\newcommand{\Ww}{\mathcal{W}}
\newcommand{\Ss}{\mathcal{S}}
\newcommand{\NN}{\mathbb{N}}
\newcommand{\RR}{\mathbb{R}}
\newcommand{\EE}{\mathbb{E}}
\newcommand{\PP}{\mathbb{P}}
\newcommand{\QQ}{\mathbb{Q}}
\newcommand{\D}{\mathrm{d}}

\newcommand{\E}{\mathrm{e}}

\newcommand{\eps}{\varepsilon}

\newcommand{\Hh}{\mathcal{H}}

\newcommand{\Ff}{\mathcal{F}}
\newcommand{\df}{\mathfrak{d}}
\newcommand{\wf}{\mathfrak{w}}
\newcommand{\Df}{\mathrm{D}}
\newcommand{\Af}{\mathrm{A}}
\newcommand{\Wf}{\mathrm{W}}
\newcommand{\Pf}{\mathrm{P}}
\newcommand{\Mf}{\mathfrak{M}}

\mathtoolsset{showonlyrefs=true}
\DeclareMathOperator*{\argmax}{\arg\!\max}
\DeclareMathOperator*{\argmin}{\arg\!\min}
\newcommand{\Pab}{\Pp_{a,b}^{d}}

\newcommand{\tens} {\mathbin{\mathop{\otimes}}}
\DeclarePairedDelimiter{\norm}{\lVert}{\rVert}
\newcommand{\kldiv}{\mathrm{KL}\kldivx}
\DeclarePairedDelimiterX{\kldivx}[2]{(}{)}{#1\delimsize\|#2}

\begin{document}

\title{Market regime classification with signatures}
\author{Paul Bilokon}
\address{The Thalesians and Department of Mathematics, Imperial College London}
\email{paul@thalesians.com}
\author{Antoine Jacquier}
\address{Department of Mathematics, Imperial College London and the Alan Turing Institute}
\email{a.jacquier@imperial.ac.uk}
\author{Conor McIndoe}
\address{Department of Mathematics, Imperial College London}
\email{conormcindoe1@gmail.com}
\thanks{
The \texttt{Python} code related to the results in this paper is available at
\href{https://github.com/mcindoe/regimedetection}{https://github.com/mcindoe/regimedetection}.}
\date{\today}
\maketitle


\begin{abstract}
We provide a data-driven algorithm to classify market regimes for time series. 
We utilise the path signature, encoding time series into easy-to-describe objects, and provide a metric structure which establishes a connection between separation of regimes and clustering of points.
\end{abstract}


\section{Introduction}

Market regimes are a clear feature of market data time series, with notions such as bull markets, bear markets, periods of calm and those of turmoil being commonplace in discussions between practitioners. As the saying goes, liquidity begets liquidity; regimes may be self-reinforcing for some time before a market shift is observed.
The dramatic market event of early 2020, brought about by the COVID-19 pandemic, brought with it high volatility and low liquidity. History does indeed seem to repeat itself: similar conditions were observed in many crashes over the last century, such as the 1929 Wall Street crash, the Black Monday event of 1987, the 2008 global crisis, and the 2015-2016 Chinese stock market crash.

The ability of an investor to recognise the underlying economic and market conditions and, ideally, to estimate the transition probabilities between market regimes, has long sought attention.
Kritzman, Page and Turkington~\cite{Kritzman12} used Markov-switching models to characterise regimes for portfolio allocation.
Jiltsov~\cite{Jiltsov20} has used Hidden Markov models on equity data to identify market clusters to capture opinions on credit risk of large banks.
This approach is partially data-driven, in that the number of regimes is initially asserted, and the resulting classifications are analysed once fitted.
Another approach, as as in~\cite{Nystrup20}, is to decide on some set of regimes and then segregate market data into the predetermined categories.

In this paper, we seek a framework which is a data-driven as possible, without specifying either the number or characteristics of the final regimes.
To do so, we make intensive use of the signature of a path, which originated in~\cite{Chen58} and was developed in the context of rough paths by Lyons and coauthors~\cite{Boedihardjo, Lyons10, Lyons07, LyonsBook, LyonsInverting}.
The path signature has recently received much attention, proving itself to be a natural language to encode time series data in a form suitable for machine learning tasks. 
The key idea of the present contribution is to use these path signatures as
points in some suitable metric space that can then be classified using a clustering algorithm.

In Section~\ref{cha:clustering}, we present the Azran and Ghahramani clustering algorithm~\cite{Azran06} and show how to apply it to finite-dimensional data.
We recall in Section~\ref{sec:ClusteringSign} the basic definitions and properties of path signatures.
Finally, in Section~\ref{sec:Application}, we show how to incorporate signatures 
as points in the clustering algorithm and how to define a suitable notion of distance between these points.
We provide a numerical example on synthetic data as a demonstration of the algorithm.

\section{Data-driven clustering}\label{cha:clustering}

We recall here the data-driven clustering algorithm by Azran and Ghahramani~\cite{Azran06} (AG-algorithm) over arbitrary metric spaces.
The algorithm is purely data-driven, in that the number and shape of clusters is left unspecified and both are suggested by the algorithm.

\subsection{The Azran-Ghahramani clustering}

We consider a given metric space with a distance function~$\df$
and a collection of points $\Xx = \{x_1, \ldots, x_n\}$.

\begin{definition}
A similarity function~$\wf:\RR_+\to\RR_+$ is a monotonically decreasing function.
Given $(\Xx, \df)$ as above, the similarity matrix is the matrix $\Wf = (\ww_{i,j})_{1\leq i,j\leq n}$
defined as $\ww_{i,j}:=\wf(\df(x_i,x_j))$.
\end{definition}

The $(i,j)$-entry of $\Wf$ is the similarity between points $x_i$ and $x_j$. Natural candidates for the similarity function include the inverse function $\wf(x) = x^{-1}$,
the squared inverse function $\wf(x) = x^{-2}$ and the Gaussian $\wf(x) = \exp(-x^2)$.
We further define the matrix~$\Pf:=\Df^{-1} \Wf$,
where~$\Df$ is the diagonal matrix in~$\Mm_{n}(\RR)$ with $\Df_{ii}: = \sum_{j = 1}^{n} \ww_{ij}$,
so that~$\Pf$ corresponds to a transition matrix.
We introduce the natural notion of cluster as a set of points that are close to each other:
\begin{definition}
\label{def:clusters}
A cluster of size $k\leq n$ is a subset $\Cc\subset \Xx$
such that $\min\limits_{x,y \in\Cc}\wf(\df(x,y)) > \max\limits_{x\in\Cc, y \in \Xx\setminus\Cc}\wf(\df(x,y))$.
\end{definition}

We consider the random walk of~$n$ particles $X_1,\ldots, X_n$, starting from $x_1,\ldots, x_n$ respectively, whose location evolves according to the homogeneous Markov chain with transition matrix~$\Pf$, so that $\Pf_{i,j}$ is the probability that the particle moves from~$x_i$ to~$x_j$ between two time steps.
Let $X_i(t)$ denote the (row vector) discrete distribution of the $i$\textsuperscript{th} particle after~$t$ steps.
Clearly $X_i(0)$ is a Dirac mass centered at~$x_i$; 
at any (discrete) time $t \geq 1$, the discrete distribution of the $i$\textsuperscript{th} particle is given by
\begin{equation}\label{eq:distribution-of-x-i}
    X_i(t) = X_i(t-1)\Pf = \cdots = X_i(0)\Pf^t = \begin{pmatrix} 0 & \cdots & 0 & 1 & 0 & \cdots & 0\end{pmatrix}\Pf^t,
\end{equation}
where the~$1$ is in the $i$\textsuperscript{th} position, 
as a placeholder for the position~$x_i$.

In a well-clustered space, where similarities between points of the same cluster are high and between points of distinct clusters are low, we expect particles to mostly remain within their clusters throughout the random walk. So, after a sufficient number of steps, the distributions of particles beginning in these well-separated clusters should be similar. It follows then by~\eqref{eq:distribution-of-x-i} that the corresponding rows of~$\Pf^t$ will be similar. Note that this establishes a correspondence between similar points in the metric space and similar rows in the matrix~$\Pf^t$.
We summarise the following important properties of the matrix~$\Pf$~\cite[Lemma 1, Lemma 2]{Azran06}:

\begin{lemma}\label{lem:properties-of-P}
If $\Wf$ is full rank, then so is $\Pf$.
The spectrum of~$\Pf$ is of the form $1=\lambda_1\geq \cdots \geq \lambda_n\geq -1$.
Let~$v_k$ be the eigenvector corresponding to~$\lambda_k$, chosen with unit norm.
Then, for any $t \geq 1$,
\begin{equation}\label{eq:computation-of-P^t}
        \Pf^{t} = \sum_{k = 1}^{n} \lambda_{k}^{t} \Af_k,
    \end{equation}
with $\Af_k := \frac{v_k v_k^{T}}{v_k^{T} \Df v_k}\Df$ idempotent and orthogonal,
$\{\Af_k\}_{k = 1}^{n}$ forming a basis of the space generated by
$\{\Pf^k\}_{k\geq 1}$.
\end{lemma}

From~\eqref{eq:computation-of-P^t}, we see that eigenvalues close to~$1$ correspond
to more stable basis elements, whereas small eigenvalues quickly shrink to zero as time evolves.
As discussed in~\cite[Assumption A1]{Ng01}, 
the special case of~$K$ separated clusters with no connections between points in different clusters 
(i.e. zero similarity between such points) corresponds to the eigenvalues satisfying 
$1=\lambda_1 = \cdots = \lambda_K>\lambda_{K+1}$.
In this case~\eqref{eq:computation-of-P^t} implies that~$\Pf^t$ converges to $\sum_{k = 1}^{K} \Af_k$
as~$t$ tends to infinity.
This motivates the following definition, which gathers the essential tools required to build the clustering algorithm:

\begin{definition}
\label{def:previous-steps-to-best-reveal}
For $t\geq 1$ and $k=1, \ldots, K$, define the~$k$\textsuperscript{th} eigengap after~$t$ steps by $\Delta_k(t) := \lambda_k^t - \lambda_{k+1}^t$ and
let $\Kk_t := \argmax_k \hspace{0.1em} \Delta_k(t)$ be the $k$-eigengap which is largest after $t$ time steps, at which point we say $k$ clusters are revealed.
For a given number of clusters, $k$, $T_k:=\{t:\Kk_t = k\}$ represents the set of all time steps at which~$k$ clusters are revealed. We call 
$t_k := \argmax_{t\in T_k}\Delta_k(t)$ the $k$-cluster revealer, at which point the~$k$ clusters are best segregated.
Finally, we call
$\Delta(t) := \max\limits_{k \in \{1, \ldots, n\}} \Delta_k(t)$ the maximal eigengap separation after~$t$ steps.
\end{definition}

We use the term $k$-clustering for a partition of~$\Xx$ into~$k$ subsets,
and say that a $k'$-clustering is better revealed than a $k$-clustering after~$t$ steps if $\Delta_{k'}(t) > \Delta_k(t)$.

The AG-algorithm suggests $k$-clusterings for values of~$k$ for which there exists some number of steps~$t$, 
at which~$k$ clusters is better revealed than any other number of clusters~$k'$. 
If, however, there exists some $k' \neq k$ for which $\Delta_k(t_k) < \Delta_{k'}(t_k)$, then $k$ is not considered a suitable number of clusters for the data. 
We are therefore interested in computing the set of time steps $\Tt := \{t_1, \ldots, t_m\}$ where $\Delta(\cdot)$ attains a local maxima, and then for each $t_i \in \Tt$ computing the number of clusters, $k_i := \Kk_{t_i}$, where $\Delta_{k_i}(t_i) = \Delta(t_i)$. Finally, the suggested $k_i$-clustering of the points is inferred by finding a $k_i$-clustering of the rows of the matrix $\Pf^{t_i}$. The value $\Delta_{k_i}(t_i)$, bounded above by 1, provides a measure of the separation of clusters in the returned partition, as motivated by the discussion before Definition~\ref{def:previous-steps-to-best-reveal}.

We summarise the process in Algorithm~\ref{alg:multiscale-k-prototypes}, deferring for now the discussion of how to compute the $k$-clustering for given $k$.

\begin{algorithm}[ht]\label{algo:MultiScale}
    \label{alg:multiscale-k-prototypes}
    \textbf{Input}: Metric space $(\Xx, \df)$, maximum number of steps $T$\\
    \textbf{Output}: Collection of suggested partitions with the corresponding eigengap separations
    \caption{Multiscale $k$-Prototypes Algorithm}
    \begin{enumerate}[(i)]
        \item Compute $\Pf$ and its spectrum $\lambda_1 \geq \ldots \geq \lambda_n$;
        \item Compute $\Delta(t)$ for $t \in \{1, \ldots, T\}$. Find the set of local maxima $\Tt := \{t_1, \ldots, t_m\}$;
        \item For each $t_i \in \Tt$, find the number of clusters $k_i$ best revealed by $t_i$ steps;
        \item For each $k_i$, compute the corresponding $k_i$-partitioning $\mathcal{I}_{k_i}$;
        \item Return the final collection $\{(\mathcal{I}_{k_1}, \Delta(t_{k_1})), \ldots, (\mathcal{I}_{k_m}, \Delta(t_{k_m}))\}$.
    \end{enumerate}
\end{algorithm}

In order to determine the $k$-clusterings~$\mathcal{I}_k$ in Algorithm~\ref{algo:MultiScale}, we make use of an algorithm which is similar to $k$-means clustering, called the $k$-prototypes algorithm (the word `prototype', borrowed from~\cite{Azran06}, refers to the distribution vectors of the particles following the random walk).
In $k$-means clustering, a distance between vectors is used to separate points into $k$ clusters. Here we have distributions, and hence a slightly different approach is suggested in~\cite{Azran06}, making use of the Kullback-Leibler divergence which we now recall.

\begin{definition}\label{def:kl-divergence}
Given two probability distributions~$\PP$ and~$\QQ$ on~$\Xx$,
the Kullback-Leibler (KL) divergence from~$\QQ$ to~$\PP$ is defined as
$$\kldiv{\PP}{\QQ} := \sum_{x \in \Xx} \PP(x) \log \left( \frac{\PP(x)}{\QQ(x)} \right )$$
\end{definition}

The Kullback-Leibler divergence is not a proper distance function since it is not symmetric,
but nevertheless gives a notion of disparity between two probability distributions,
and is hence used in the $k$-prototype algorithm~\cite{Azran06} to compare the distributions
of the particles.
As argued in~\cite{Azran06}, the usual Euclidean distance is not adapted here as it gives all elements the same weights.
For a given number of clusters~$k$ and number of steps~$t$, a suitable $k$-clustering is identified by minimising the function
$$
\sum_{j=1}^{k}\sum_{m\in\mathcal{I}_j}\kldiv{\Pf^t_m}{Q_j},
$$
over all partitioning $\mathcal{I} = (\mathcal{I}_1,\ldots, \mathcal{I}_k)$ and distributions
$(Q_1,\ldots, Q_k)$.
We refer to~\cite{Azran06} for a precise detail of the recursive algorithm to solve this non-convex optimisation problem, together with some specification about the initial starting point of the algorithm.
Their approach is summarised in Algorithm~\ref{alg:k-prototypes}.

\begin{algorithm}[ht]
    \label{alg:k-prototypes}
    \caption{$k$-Prototypes Algorithm}
    \textbf{Input}: Transition matrix $\Pf^t\in \RR^{n \times n}$, number of clusters $k$, initial matrix $Q \in \RR^{k \times n}$ of prototypes; \\
    \textbf{Initialisation}: $Q^{\textrm(old)} := Q$; \\
    \textbf{Output}: Partition $\mathcal{I}$, of size $k$, of the indices $\{1, \ldots, n\}$;\\
    \begin{enumerate}[(i)]
    \item For $j \in \{1, \ldots, k\}$, $\mathcal{I}_j^{\rm{(new)}} := \left\{ m: j = \argmin_{j \in \{1, \ldots, k\}} \kldiv[\big]{\Pf^t_m}{Q_j^{\rm{(old)}}} \right\}$;
    \item For $j \in \{1, \ldots, k\}$, define
        $Q_j^{\textrm{(new)}} := | \mathcal{I}_j^{\textrm{(new)}}|^{-1}\sum_{m \in \mathcal{I}_{j}^{\textrm{(new)}}} \Pf^t_m$;
    \item Set $Q^{\textrm{(old)}} := Q^{\textrm{(new)}}$ and return to (i) until convergence or stop criterion.
    \end{enumerate}
\end{algorithm}

Note that step (ii) of Algorithm~\ref{alg:k-prototypes} is not well-defined if any of the partition elements $\mathcal{I}_j$ are empty. Our approach in this case is to first compute the prototypes $Q_j$ for all $j$ where $| \mathcal{I}_j | > 0$, and then iteratively define the remaining prototypes $Q_m$ to be the row of the transition matrix $\Pf^t$ for which the minimum KL-divergence to all currently-defined prototypes is maximised. This ensures the space remains well-covered by the cluster prototypes; a similar technique is discussed to initialise the matrix of prototypes in~\cite[Section 4.2]{Azran06}


\subsection{Gaussian Clouds}

We now turn to applications of the AG-algorithm,
and begin with an example in~$\RR^2$, equipped with the Euclidean distance.
For $\mm\in \RR^2$, $\sigma>0$, $n \in \NN$, we denote by~$\Xx_{\mm}^{\sigma}(n)$ a Gaussian Cloud of size~$n$,  centre~$\mm$ and variance~$\sigma^2$, 
that is a collection $\{a_1, \ldots, a_n\}$ where each $a_i$ is Gaussian with mean~$\mm$ and covariance matrix $\sigma^2 I_2$.
Figure~\ref{fig:gaussian-clouds-points} shows the generated points, with four cluster centres, 
each cluster consisting of~$100$ points.

\begin{figure}[ht]
    \centering
        \includegraphics[scale=0.5]{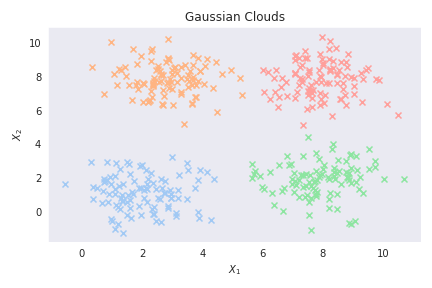}
    \caption{Gaussian Clouds where $\sigma$ is set equal to one for all four clusters and the centers~$\mm$ are 
$(2, 1)$, $(3, 8)$, $(8, 2)$, $(8, 8)$.}
    \label{fig:gaussian-clouds-points}
\end{figure}

We make use of the Gaussian similarity function:
\begin{equation}\label{eq:gaussian-similarity}
\wf^\xi(x) := \exp \left( - \frac{x}{\xi^2} \right).
\end{equation}
The optimal choice of~$\xi$ is not entirely obvious and, following~\cite{Azran06}, we choose it to be the $1\%$ low-value quantile of the non-zero distances in the space.
In Figure~\ref{fig:gaussian-clouds-eigengap-separations}, we demonstrate in the left-hand plot the curves $\Delta_k(\cdot)$ for some small values of~$k$. The right-hand plot is the curve $\Delta(\cdot)$, which we recall as the corresponding maximum over all~$k$. 
The local maxima of this latter curve represents points at which some clustering is best revealed. 
\begin{figure}[hb!]
    \centering
        \includegraphics[scale=0.5]{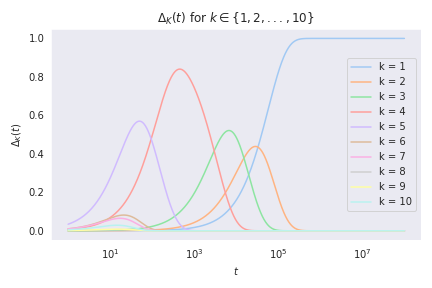}
        \includegraphics[scale=0.5]{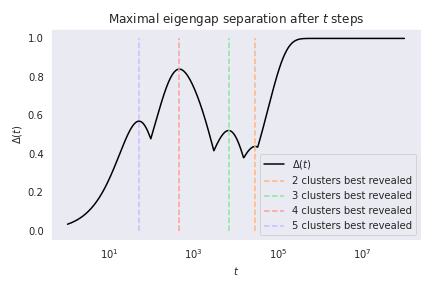}
        \caption{Gaussian Clouds - $\Delta_k(t)$ and $\Delta(t)$.}
        \label{fig:gaussian-clouds-eigengap-separations}
\end{figure}

The suggested 3-, 4- and 5-clusterings are displayed in Figure~\ref{fig:gaussian-clouds-suggested-clusterings}. 
The suitability of the resulting partitions, as indicated by the eigengap separation, suggests that the preferred partition of the points is the 4-clustering, followed by the 5-clustering and finally the 3-clustering. The recommended 4-clustering is almost the same but not identical to the original problem specification. 

\begin{figure}[ht]
    \centering
        \includegraphics[scale=0.5]{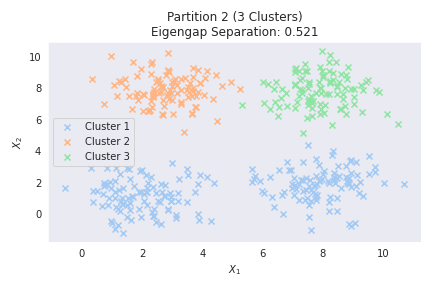}
        \includegraphics[scale=0.5]{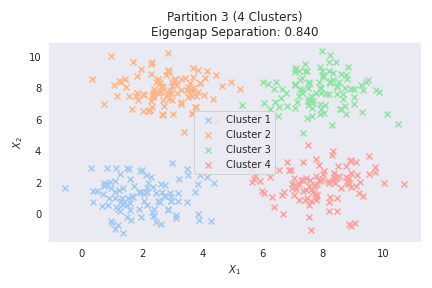}
        \includegraphics[scale=0.5]{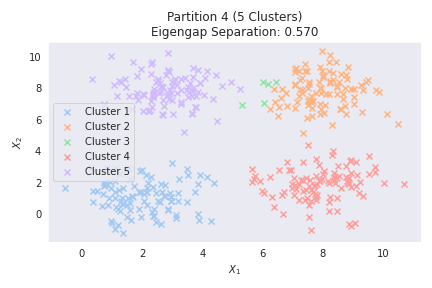}
        \caption{Gaussian Clouds - 3 suggested clusterings}
        \label{fig:gaussian-clouds-suggested-clusterings}
\end{figure}

\section{Clustering paths via signatures}\label{sec:ClusteringSign}
We now move on to the key part of the paper, demonstrating how the AG-algorithm above can be applied to time series.
The obvious two hurdles to overcome are the infinite-dimension feature of time series
and a suitable notion of distance.
To do so, we summarise in Section~\ref{sec:Signature} the information contained in a series in its so-called signature,
and show then how this helps us extend the AG-algorithm.

\subsection{An overview of signatures}\label{sec:Signature}
We provide an overview of signatures of paths for completeness.
They date back to Chen~\cite{Chen58}, but have received widespread investigation in the context of rough paths~\cite{Boedihardjo, Lyons10}.
Signatures can be defined for large classes of functions of bounded variation~\cite{FrizVictoir},
but we restrict our analysis for simplicity to the case of piecewise smooth functions.
A path~$\gamma$ is a continuous map from $[a,b]$ to~$\RR^d$ ($d\in\mathbb{N}$).
If the map $t \mapsto \gamma_t$ is differentiable,
the integral of a function $f:\RR^d \to \RR^p$
along~$\gamma$ is classically defined by
$\int_a^b f(\gamma_t) \D{\gamma_t} := \int_a^b f(\gamma_t) \dot{\gamma}_t \D{t}$.
In the case of piecewise differentiable curves~$\Pab$ from $[a,b]$ to~$\RR^d$, we may extend this to
$$
\int_{a}^{b} f(\gamma_t) \D{\gamma_t} := \int_{x_1}^{x_2} f(\gamma_t) \D{\gamma_t}  + \int_{x_2}^{x_3} f(\gamma_t) \D{\gamma_t} + \ldots + \int_{x_{n-1}}^{x_n} f(\gamma_t) \D{\gamma_t}
\in\RR^p,
$$
where $a = x_1 < x_2 < \ldots < x_n = b$ is a partition of $[a,b]$ with~$\gamma$ differentiable over each interval $(x_i, x_{i+1})$.

\begin{definition}\label{def:k-fold-iterated-integral}
Let $\gamma\in\Pab$.
For $k \in \NN$ and $\ii = (i_1, \ldots, i_k) \in \{1, \ldots, d\}^k$, we define recursively, for $t\in (a,b]$,
\begin{align*}
\Ss(\gamma)_{a,t}^{i} & := \int_{(a,t]}\D\gamma_s^i, \quad \text{for } i=1,\ldots, d,\\
\Ss(\gamma)_{a,t}^{\ii} & :=  \int_{a < t_1 < \ldots < t_k < t} \D{\gamma_{t_1}^{i_1}} \cdots \D{\gamma_{t_k}^{i_k}} = \int_{a<s<t} \Ss(\gamma)_{a,s}^{i_1, \ldots, i_{k-1}} \D{\gamma_s^{i_k}}.
\end{align*}
\end{definition}

The path signature is then defined as the collection of all such iterated integrals.
This can be stated in a more elegant way though, via the notion of alphabets.
For a given alphabet~$A$, i.e. a (finite or not) sequence of elements,
a word~$\xx$ of length $|\xx|\in\NN$ is an ordered sequence $x_1 x_2 \cdots x_{|x|}$
with $x_i \in A$ for $i \in \{1, \cdots, |x|\}$.
We denote by~$\Ww(A)$ the set of all possible words, including the empty word~$\eps$ (of length zero).
For example, the set of words of length 3 on~$A := \{a, b\}$ is
$ \{aaa, aab, aba, abb, baa, bab, bba, bbb\} $, while
the set of all words on~$A$ is
$\Ww(A) = \{\eps, a, b, aa, ab, ba, bb, aaa, aab, \ldots\}$.
If an ordering~$<$ exists on the alphabet, we can extend it recusively to~$\Ww(A)$:
the concatenation of $\xx, \yy\in\Ww(A)$ is defined to be the word $\xx\yy := x_1,\cdots,x_{|\xx|}, y_1,\cdots,y_{|\yy|}$;
then by setting $\eps < \xx$ whenever $|\xx|>0$,
we say, for any $a, b \in A$, that $a\xx < b\yy$ if either $a < b$ (in~$A$), or $a = b$ and $\xx < \yy$.

\begin{definition}\label{def:path-signature}
For $\gamma\in\Pab$, the path signature $\Ss(\gamma)_{a,b}$ is the collection of iterated integrals
$\Ss(\gamma)_{a,b}^{\ii}$ with $\ii\in \Ww(\{1, \ldots, d\})$ and $\Ss(\gamma)_{a,b}^{\eps}:=1$.
The truncated signature $\Ss(\gamma)_{a,b}^{\leq n}$ up to level~$n\in\NN$
is the restriction of $\Ss(\gamma)_{a,b}^{\ii}$ to $|\ii|\leq n$. The terms in the signature appear in the same order as the corresponding words in the natural ordered alphabet:
$$
\Ss(\gamma)_{a,b} := \left(  \Ss(\gamma)_{a,b}^{\eps}, \Ss(\gamma)_{a,b}^{1}, \ldots, \Ss(\gamma)_{a,b}^{d}, \Ss(\gamma)_{a,b}^{11}, \Ss(\gamma)_{a,b}^{12}, \ldots \right)
$$
\end{definition}

The following proposition demonstrates that the starting point of the path is lost when projecting to the signature: indeed, the signature is invariant to translations of the original path.

\begin{proposition}[Translation invariance]
\label{prop:translation-invariance-of-path-signature}
For any $\gamma\in\Pab$, $\cc\in \RR^d$ and $\ii \in \Ww(\{1, \ldots, d\})$, we have
$$\Ss(\gamma + \cc)^{\ii}_{a, b} = \Ss(\gamma)^{\ii}_{a, b}$$
\end{proposition}

\begin{proof}
    Write $\cc = (c^1, \ldots, c^d)$, $\widetilde{\gamma} := \gamma + \cc$, and take $i \in \{1, \ldots d\}$. We have $\widetilde{\gamma}^i_t = \gamma^i_t + c^i$, 
hence 
    \begin{equation*}
        S(\widetilde{\gamma})^{i}_{a, b} = \int_a^b \diff{\widetilde{\gamma}_t^i} = \int_a^b \frac{\D \widetilde{\gamma}_t^i}{\D t} \diff{t}
 = \int_a^b \frac{\D \gamma_t^i}{\D t} \diff{t} = S(\gamma)^i_{a,b},
    \end{equation*}
    which shows the result for any word of length one. Now suppose the result holds for any $k$-length word. Let $i_1 \cdots i_k i_{k+1}$ be a $(k+1)$-length word; we have
    \begin{equation*}
        S(\widetilde{\gamma})^{i_1, \ldots, i_k, i_{k+1}}_{a,b} = \int_a^b S(\widetilde{\gamma})_{a, b}^{i_1, \ldots, i_k} \hspace{0.2em} \frac{\D \widetilde{\gamma}^{i_{k+1}}_t}{\D t} \diff{t} = \int_a^b S(\gamma)_{a, b}^{i_1, \ldots, i_k} \frac{\D \gamma_t^{i_{k+1}}}{\D t} \diff{t} = S(\gamma)_{a,b}^{i_1, \ldots, i_k, i_{k+1}},
    \end{equation*}
so the terms of the signatures of each path agree for any word in $\mathcal{W}(\{1, \ldots, d\})$, 
hence for the entire signature.
\end{proof}

\subsubsection{Logsignatures}
\label{sec:log-signature}
The logsignature is a more concise representation of the information present in a signature.
We introduce the vector space~$\Vv$ of non-commutative formal power series on a basis of symbols $B = \{\E_1, \ldots, \E_r\}$,
that is the set of elements of the form $\sum_{\ww\in\Ww(B)} \lambda_{\ww}\ww$, $\lambda_{\ww} \in \mathbb{R}$, with possibly many null coefficients, where for now $\lambda \ww$ is interpreted as a formal symbol rather than a multiplication.
We endow~$\Vv$ with a vector space structure, namely for any $\lambda, \mu \in \RR$
and $\ww\in \Ww(B)$,
both $\lambda (\mu \ww) := (\lambda \mu) \ww$ and $\lambda \ww + \mu \ww := (\lambda + \mu) \ww$ hold.
We define $\tens: \Vv \times \Vv \to \Vv$ as
the concatenation $\ww \tens \vv := \ww\vv$, so that~$\Vv$ acquires the structure of an algebra.
For a path~$\gamma\in\Pab$, we may identify its signature with the non-commutative formal power series
\begin{equation}\label{eq:FormalPower}
\Ss(\gamma)_{a, b} = 1 + \Ss(\gamma)_{a, b}^1 \E_1 + \ldots + \Ss(\gamma)_{a,b}^d \E_d + \Ss(\gamma)_{a,b}^{1, 1} \E_1 \E_1 + \Ss(\gamma)_{a,b}^{1,2} \E_1 \E_2 + \ldots
\end{equation}
The $=$ sign is a slight abuse of notation, but is justified as a one-to-one correspondence between signatures of $d$-dimensional paths and the non-commutative formal power series in $d$ letters.
Given a power series $\omega \in \Vv$, where the coefficient~$\lambda_0$ of the empty word
is not null, we define
\begin{equation}\label{eq:logsignature-defn}
\log \omega := \log(\lambda_0) + \sum_{n \geq 1} \frac{(-1)^n}{n} \left( 1 - \frac{\omega}{\lambda_0}\right)^{\tens n},
\end{equation}
and thus deduce the logsignature $\log \Ss(\gamma)_{a,b}$ of~$\gamma$ of a path $\gamma\in\Pab$.

From Definition~\ref{def:k-fold-iterated-integral}, the level-one term~$\Ss(\gamma)^i_{a,b}$
corresponds to the displacement of the coordinate path $\gamma^i$ between times~$a$ and~$b$.
In fact, this $i$\textsuperscript{th} displacement fully determines the $k$-fold iterated integral
over the $i$\textsuperscript{th} index:
\begin{proposition}\label{prop:signature-of-word-of-one-letter}
For $\gamma\in\Pab$, the identity
$\Ss(\gamma)^{\overbrace{i, \ldots, i}^{k \,\rm{times}}}_{a, b} = (\gamma^i_b - \gamma^i_a)^k / k!$
holds for any $i \in \{1, \ldots, d\}$, $k \in \NN$.
\end{proposition}

\begin{proof}
Let $i \in \{1, \ldots, d\}$. The $k = 1$ case is trivial. Suppose the result holds for $k \in \NN$; by induction
    \begin{equation*}
        \Ss(\gamma)^{\overbrace{i, \ldots, i}^{k+1 \hspace{0.15em} \rm{times}}}_{a, b} = \int_a^b \frac{(\gamma^i_t - \gamma^i_a)^k}{k!} \hspace{0.2em} \frac{\D \gamma^i_t}{\D t} \D{t} = \left[ \frac{(\gamma^i_t - \gamma^i_a)^{k+1}}{(k+1)!} \right]_a^b = \frac{(\gamma^i_b - \gamma^i_a)^{k+1}}{(k+1)!}.
    \end{equation*}
\end{proof}

This shows that the signature of a one-dimensional path is completely determined by its displacement. 
It should be noted however that the time series of price paths, as considered in this paper, are not one-dimensional but instead are two-dimensional processes with time in the first coordinate, 
namely
$\{(t_1, x_1), \ldots, (t_n, x_n)\}$ instead of $\{x_1, \ldots, x_n\}$.

Another important result connects second-order signature terms to the product of first-order terms. 
We provide an understanding of this result in the case where the curve is the concatenation of paths which are piecewise monotone, and postpone its proof to Appendix~\ref{app:Proof}.

\begin{lemma}\label{lem:sum-of-second-order-signature-terms}
The equality
$\Ss(\gamma)^{i, j}_{a, b} + \Ss(\gamma)^{j, i}_{a, b} = \Ss(\gamma)^i_{a,b} \Ss(\gamma)^j_{a,b}$
holds for any path $\gamma\in\Pab$ and $1\leq i,j \leq d$.
\end{lemma}

\begin{figure}[ht]
    \centering
        \includegraphics[scale=0.4]{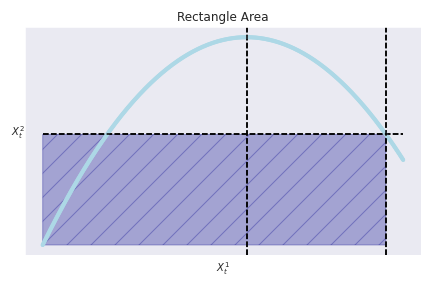}
        \includegraphics[scale=0.4]{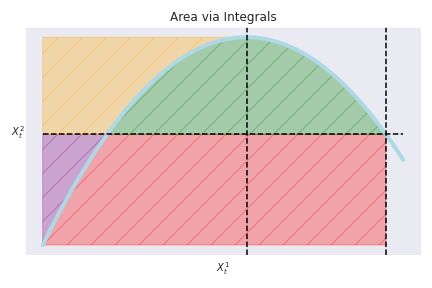}
    \caption{Inductive step for Lemma~\ref{lem:sum-of-second-order-signature-terms} - two methods to count the product of the displacements, depicted as the blue rectangle}
    \label{fig:area-via-integrals}
\end{figure}

A geometric interpretation of this result may be seen in Figure~\ref{fig:area-via-integrals}. The integral $\Ss(\gamma)^{1,2}$ for example may be considered as the sum of the signed areas over each section where the curve is monotone. 
The statement then reads that the signed area of the blue rectangle in the left image may be constructed as follows: take the integral of $\gamma^2 \D{\gamma^1}$, to contribute the positive green and red sections, then the integral of $\gamma^1 \D{\gamma^2}$ from zero until the maxima (with respect to the $\gamma^1$ coordinate, the left-most black dotted line) contributes positive purple and yellow areas; the final integral from the maxima to the end point (second black dotted line) contributes negative green and negative yellow areas. The resulting sum of all these signed areas is the area of the blue rectangle.

\begin{remark}
Lemma~\ref{lem:sum-of-second-order-signature-terms} highlights that there is some redundancy in the vector representation of the signature which we have so far discussed.
If the terms $\Ss(\gamma)^i, \Ss(\gamma)^j$ and $\Ss(\gamma)^{i, j}$ are known, then $\Ss(\gamma)^{j, i}$ may be inferred without an explicit record in the vector.
In the same way, the terms corresponding to words in one letter have the representation of Proposition~\ref{prop:signature-of-word-of-one-letter}; we note that only one of these terms is required for each letter in order to compute all such terms.
\end{remark}

Let us recall the notion of a Lie bracket operation $[\cdot, \cdot]$ induced by~$\tens$:
$[x,y] := x \tens y - y \tens x$ for $x, y \in \Vv$.
Combining~\eqref{eq:FormalPower} and~\eqref{eq:logsignature-defn} yields, for any path $\gamma\in\Pab$,
    \begin{equation}
        \label{eq:logsignature-expansion}
        \log \Ss(\gamma)_{a,b} = \sum_{k \geq 1} \sum_{i_1, \ldots, i_k \in \{1, \ldots, d\}} 
\lambda_{i_1, \ldots, i_k} \Big[e_{i_1}, \big[e_{i_2}, \ldots, [e_{i_{k-1}}, e_{i_k}] \ldots\big]\Big].
    \end{equation}
A vector representation of the logsignature therefore only requires entries corresponding to each of the basis elements of the form appearing in~\eqref{eq:logsignature-expansion}. 
This has considerably fewer terms (up to a given level) than the full signature in Definition~\ref{def:path-signature}. 
A five-dimensional path for example, up to the third level, has~$155$ terms in its signature and~$55$ in its logsignature (see~\cite{Reizenstein17} for formulae on the sizes of signatures). 
This makes the logsignature particularly enticing from a computational point of view, 
not just for the reduction of  required working memory, 
but also because convergence time of algorithms will benefit from the removal of redundancy in the feature set.

\subsection{Signature of data points}\label{sec:signature-of-data-points}
Data, however, is not continuously observed and instead discrete observations $\{(t_i, x_i)\}_{i=1,\ldots, n}$ of values~$x_i$ at times~$t_i$ are more realistic.
In order to associate a signature to this set of points, some interpolation is required.
Among several approaches proposed in the literature, we present here the piecewise linear and the rectilinear interpolations used by Levin, Lyons and Ni~\cite{Levin13}.

\begin{definition}[Path interpolation]
Let $\Xx = \{(t_1, x_1), \ldots, (t_n, x_n)\}$ be a set of observations.
\begin{itemize}
\item The piecewise linear interpolation  $X: [t_1, t_n] \to \RR$ of~$\Xx$ is defined as
$$
X(t) := \sum_{i = 1}^{n-1} \left[x_i + (x_{i+1}-x_i)\frac{t - t_i}{t_{i+1}-t_i}\right]\ind_{\{t \in [t_i, t_{i+1})\}} + x_n \ind_{\{t = t_n\}};
$$
\item The rectilinear interpolation path $X':[t_1, t_n] \to \RR$ of~$\Xx$ is given by
$$
X'(t) := \sum_{i = 1}^{n-1} x_i \ind_{\{t \in [t_i, t_{i+1})\}} + x_n \ind_{\{t = t_n\}}.
$$
\end{itemize}
\end{definition}
Figure~\ref{fig:data-interpolation} shows the behaviour of these two interpolations over the set of points $\{(0, 8), (2,0), (3, 12), (6, 14)\}$:

\begin{figure}[ht]
    \centering
        \includegraphics[scale=0.4]{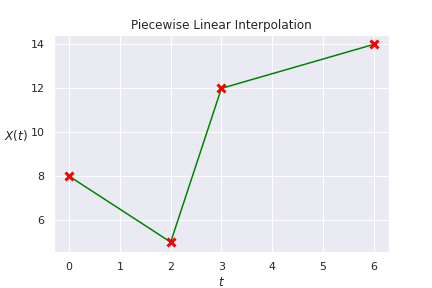}
        \includegraphics[scale=0.4]{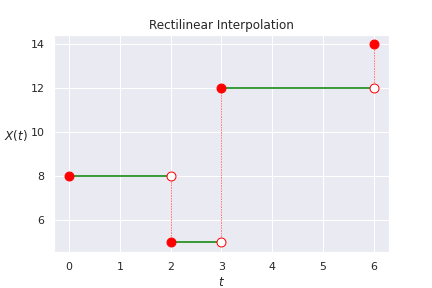}
    \caption{Piecewise linear and rectilinear interpolations}
    \label{fig:data-interpolation}
\end{figure}

In Section~\ref{sec:Application}, we will be using the piecewise linear interpolation of data points to form piecewise differentiable curves over which we may compute signatures. With this in place we may demonstrate how the AG-algorithm can be used to cluster market data time series.

\section{Numerical analysis}\label{sec:Application}

We now work towards the clustering of market time series data.
We start with simulated price paths following the Black-Scholes dynamics
\begin{equation}\label{eq:gbm-solution}
S_t = S_0 \exp \left\{ \left(\mu - \frac{\sigma^2}{2}\right)t + \sigma W_t \right\},
\end{equation}
where~$(W_t)_{t\geq 0}$ is a standard Brownian motion, and $S_0=1$ so the paths represent returns.
Here, a regime corresponds to a choice of $(\mu, \sigma)$
within a possible set of parameters $\Mf$.
We demonstrate the clustering of Brownian paths by selecting four regimes according to the parameters in Table~\ref{tab:brownian-paths-parameters}.

\begin{table}[ht]
    \centering
\begin{tabular}{|c|cccc|}
\hline
 & Regime 1 & Regime 2 & Regime 3 & Regime 4\\
\hline
    $\mu$ & 5\% & 5\% & 2\% & 2\% \\
    $\sigma$ & 10\% & 20\% & 10\% & 20\% \\
\hline
\end{tabular}
\vspace{0.1cm}
    \caption{Space $\Mf$ of parameters for the Brownian paths in Figure~\ref{fig:synthetic-paths-eigengaps}}
    \label{tab:brownian-paths-parameters}
\end{table}
\subsection{Regime points and point elements}

\begin{figure}[b]
    \centering
        \includegraphics[scale=0.55]{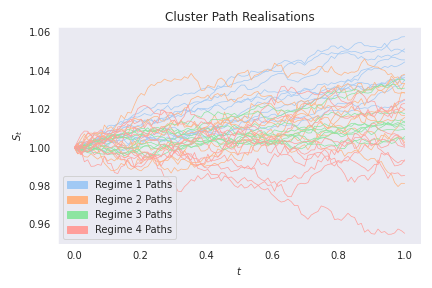}
        \includegraphics[scale=0.55]{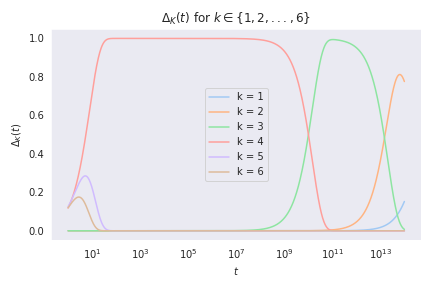}
    \caption{Synthetic paths - generated paths and eigengaps plot}
    \label{fig:synthetic-paths-eigengaps}
\end{figure}

In this setting, we will understand a point as a collection of path signatures.
The component paths are thought of as being returns of some collection of securities over some time horizon, whose prices evolve according to the dynamics of~\eqref{eq:gbm-solution} with a common regime $(\mu, \sigma)$ from Table~\ref{tab:brownian-paths-parameters}.
The paths are mapped to the interval $[0,1]$. We then compute signatures, truncated to some level, and scale each of the level-$k$ terms by~$k!$.
By repeated application
of Gronwall's Lemma~\cite[Lemma 3.2]{FrizVictoir},
the iterated integral of level~$k$ of a bounded variation path 
is equal to its $1$-norm divided by~$k!$,
hence our scaling ensures that each level is comparable to the others.
Algorithm~\ref{alg:regime-point} outlines the construction of the set of points
$\{\{\Xx^{\mu, \sigma}_i\}_{i=1,\ldots, k}\}_{(\mu,\sigma) \in \Mf}$.

\begin{algorithm}
\caption{Generation of a regime-point}
\textbf{Input}: Number of paths $n$ to generate per point, signature truncation depth $l$, number of divisions $m$ of the interval $[0,1]$, parameters $(\mu, \sigma)\in\Mf$.\\
\textbf{Output}: Single metric space point.
\begin{enumerate}
\item For each $i \in \{1, \ldots, n\}$, simulate a path
$\mathrm{S}_i = (S_i^1, \ldots, S_i^m)$, of length m according to~\eqref{eq:gbm-solution};
\item For each $i$, time-augment the path to obtain a two-dimensional path $\{(1/m, S_i^1), ~\ldots, ~(1, S_i^m)\}$;
\item Construct the piecewise-differentiable function $\widetilde{S}_i: [0,1] \to \RR^2$ by linear interpolation;
\item Let $x_i = \Ss(\widetilde{S}_i)_{0,1}^{\leq l}$ be the signature transform of the augmented path
$(\widetilde{S}_i)$ up to level~$l$;
\item Return $\Xx^{\mu, \sigma} := \{x_1, \ldots, x_n\}$
\end{enumerate}
\label{alg:regime-point}
\end{algorithm}

\subsection{Distance between collections of signatures}

Consider the points $\Xx = \Xx^{\mu, \sigma}$ and $\widetilde{\Xx} = \widetilde{\Xx}^{\widetilde{\mu}, \widetilde{\sigma}}$, generated according to Algorithm~\ref{alg:regime-point}. We may think of the regime parameters $(\mu, \sigma)$,  $(\widetilde{\mu}, \widetilde{\sigma})$, as inducing distributions~$\mathtt{P}$
and~$\widetilde{\mathtt{P}}$ of path signatures.
The expected distance between the two collections may then be understood as a distance between these two distributions.
Given a collection $X = \{x_1, \ldots, x_n\}$ of independent samples from a population~$p$
and a collection $Y = \{y_1, \ldots, y_m\}$ of samples from a population~$\widetilde{p}$,
the maximum mean discrepancy test~\cite{Gretton08} is a two-sample hypothesis test used to determine whether there is sufficient evidence at some significance level to reject the null hypothesis  $p = \widetilde{p}$.
In order to define this statistic, we first recall the definition of a reproducing kernel Hilbert space (RKHS):

\begin{definition}\label{def:reproducing-kernel-hilbert-space}
Let $\Xf$ be a set and $\Hh$ a Hilbert space of functions from~$\Xf$ to~$\RR$
endowed with an inner product~$\langle\cdot,\cdot\rangle$.
For each $x \in \Xf$, the evaluation functional $\Ll_x: \Hh \to \RR$ is defined by
$\Ll_{x} f := f(x)$.
We call~$\Hh$ a reproducing kernel Hilbert space (RKHS) if
$\Ll_x$ is continuous for every $x \in \Xf$.
In that case, for each $x \in \Xf$, there exists $K_x \in \Hh$ such that
$\Ll_x f = \langle f, K_x \rangle$, for any $f \in \Hh$
and the function $k: \Xf \times \Xf \to \RR$ such that $k(x, y) := \langle K_x, K_y \rangle$
is called the reproducing kernel for~$\Hh$.
Finally, $\Hh$ is said to be universal~\cite{Steinwart01} if $k(x, \cdot)$ is continuous for all $x$ and $\Hh$ is dense in~$\Cc(\Xf)$, the space of continuous functions from~$\Xf$ to~$\RR$.
\end{definition}

We have the following formulation of the maximum mean discrepancy statistic~\cite{Gretton08}:

\begin{definition}
Let~$\Ff$ be a class of functions from~$\Xf$ to~$\RR$, and let~$\PP$ and~$\QQ$ two distributions on~$\Xf$.
The Maximum Mean Discrepancy of~$\PP, \QQ$ over~$\Ff$ is defined as
\begin{equation}\label{eq:mmd}
\MMD^{\Ff}[\PP, \QQ] := \sup\limits_{f \in \Ff} \Big\{\EE_{X \sim \PP}[f(X)] - \EE_{y \sim \QQ}[f(Y)] \Big\}.
\end{equation}
\end{definition}

If, instead of observing~$\PP$ and~$\QQ$, we have independent observations $X = \{x_1, \ldots, x_m\}$ and $Y = \{y_1, \ldots, y_n\}$ from~$\PP$ and~$\QQ$ respectively,
then an empirical estimate of the $\MMD$ is given by
$$
\MMD^{\Ff}[X, Y] := \sup\limits_{f \in \Ff} \left\{\frac{1}{m} \sum_{i = 1}^{m} f(x_i) - \frac{1}{n} \sum_{i = 1}^n f(y_i) \right\}.
$$
Let $\Hh$ be a universal reproducing kernel Hilbert space with associated kernel~$k$
and~$\Ff$ the unit ball in~$\Hh$.
Then the empirical estimate can be computed in terms of the kernel $k(\cdot, \cdot)$ as
\begin{equation}
\MMD^{\Ff}[X,Y] =
\left[ \frac{1}{m^2} \sum_{i_1, i_2 = 1}^m k(x_{i_1}, x_{i_2})
 - \frac{2}{mn} \sum_{i = 1}^{m} \sum_{j = 1}^{n} k(x_i, y_j) + \frac{1}{n^2} \sum_{j_1, j_2 = 1}^n k(y_{j_1}, y_{j_2}) \right]^{1/2}.
        \label{eq:mmd-estimate}
    \end{equation}

We will use in particular the reproducing kernel Hilbert space induced by the Gaussian kernel (shown to be universal in \cite{Steinwart01})
\begin{equation}\label{eq:gaussian-kernel}
    k^{\sigma}(x, y) := \exp \left\{-\frac{\norm{x-y}^2}{2\sigma^2} \right\},
\end{equation}
on compact subsets of $\RR^d$, for fixed $\sigma>0$.
Note that the maximum mean discrepancy is not a distance function,
since~\eqref{eq:mmd} depends on the order of~$\PP$ and~$\QQ$.
Nevertheless, if the induced Hilbert space~$\Hh$ is a universal RKHS,
then $\MMD^{\Ff}[\PP, \QQ] = 0$ if and only if $\PP= \QQ$,
as proved in~\cite{Gretton08}.
Furthermore, the estimate~\eqref{eq:mmd-estimate} is symmetric.
We show below that this choice of metric is sufficient to allow for the classification of regimes from Brownian paths.

\subsection{Results}

\begin{figure}[ht]
    \centering
        \includegraphics[scale=0.5]{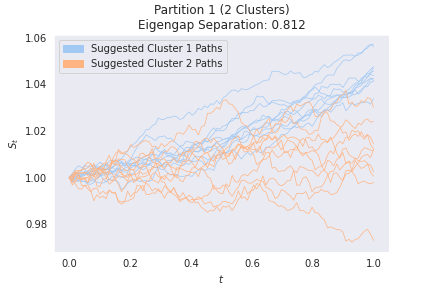}
        \includegraphics[scale=0.5]{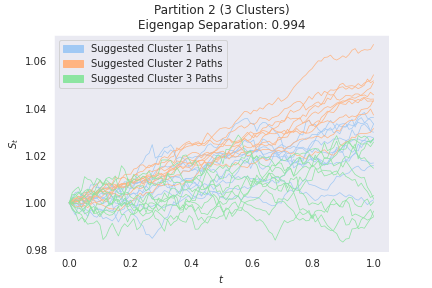}
        \includegraphics[scale=0.5]{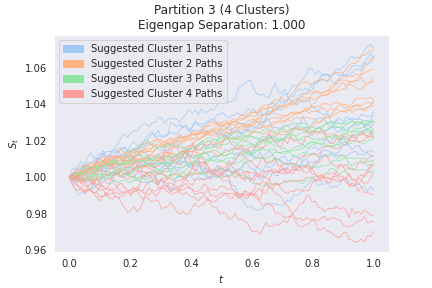}
    \caption{Synthetic paths example - suggested path clusterings}
    \label{fig:synthetic-paths-best-clusterings}
\end{figure}

For each regime in Table~\ref{tab:brownian-paths-parameters},
we generate 10 samples according to Algorithm~\ref{alg:regime-point},
each consisting of~$40$ paths, each path coming from a uniform partition of $[0,1]$ with 100 time steps.
A subset of the resulting paths (with~$10$ paths per regime) is shown
in Figure~\ref{fig:synthetic-paths-eigengaps} (left).
For each point in the space, the signatures up to level 3 are computed. The second- and third-order terms are scaled by $2!$ and $3!$ respectively.
The similarity function used was that of Equation~\eqref{eq:gaussian-similarity}.
From this setup, we may proceed with a similar analysis as in the Gaussian Clouds examples.
The maximal eigengap separation plot is presented in Figure~\ref{fig:synthetic-paths-eigengaps} (right).
As in the previous example, this structure is deemed successful at clustering the underlying points in the sense that the nontrivial clustering with the highest eigengap separation is the 4-clustering, presented in Figure~\ref{fig:synthetic-paths-best-clusterings}.
For comparison, the other suggested clusterings are also presented.
Let $\mathcal{R}_1, \ldots, \mathcal{R}_4$ denote the point indices corresponding to the regimes of Table~\ref{tab:brownian-paths-parameters}, with the same numberings. The clusters shown in Figure~\ref{fig:synthetic-paths-best-clusterings} are made precise in Table~\ref{tab:synthetic-data-cluster-partitions}. In all of the suggested partitions, the clusters are preserved, with several clusters being combined to form larger suggested clusters in the $k$-clusterings with $k < 4$.

\begin{table}[ht]
\begin{tabular}{ccc}
\hline
           & Partition                                                              & Eigengap Separation \\ \hline
2 Clusters & $\mathcal{R}_1$, $\mathcal{R}_2 \cup \mathcal{R}_3 \cup \mathcal{R}_4$ & 0.8123                  \\
3 Clusters & $\mathcal{R}_1$, $\mathcal{R}_2 \cup \mathcal{R}_3$, $\mathcal{R}_4$   & 0.9941                  \\
4 Clusters & $\mathcal{R}_1$, $\mathcal{R}_2$, $\mathcal{R}_3$, $\mathcal{R}_4$     & 1.000                   \\ \hline
\end{tabular}
\vspace{0.1cm}
\caption{Synthetic data - suggested clusters}
\label{tab:synthetic-data-cluster-partitions}
\end{table}


\appendix
\section{Proof of Lemma~\ref{lem:sum-of-second-order-signature-terms}}\label{app:Proof}

\small{
It suffices to prove the result for curves with $\gamma^i_0 = \gamma^j_0 = 0$.
Indeed, then if $\widetilde{\gamma}\in\Pab$ is the translation of $\gamma$ having $\widetilde{\gamma}^i_0 = \widetilde{\gamma}^j_0 = 0$, by Proposition~\ref{prop:translation-invariance-of-path-signature} (translation invariance) we have
$$
\Ss(\gamma)^{i, j}_{a,b} + \Ss(\gamma)^{j, i}_{a,b} = \Ss(\widetilde{\gamma})^{i, j}_{a,b} + \Ss(\widetilde{\gamma})^{j, i}_{a,b} = \Ss(\widetilde{\gamma})^{i}_{a,b} \Ss(\widetilde{\gamma})^{j}_{a,b} = \Ss(\gamma)^{i}_{a,b}\Ss(\gamma)^{j}_{a,b}.
$$

So assume without loss of generality that $\gamma^i_0 = \gamma^j_0 = 0$.
The result is clear when the function is monotone. 
We refer to Figure~\ref{fig:monotone-function} below. 
The term on the right-hand side is the area of the bounding rectangle, and the term on the left is the sum of the shaded areas, which are the integrals in both directions.
\begin{figure}[h!]
    \centering
    \includegraphics[scale=0.4]{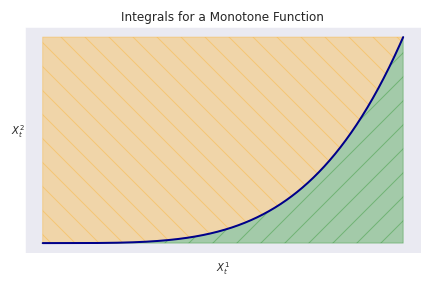}
    \caption{Monotone function step (base case) for Lemma~\ref{lem:sum-of-second-order-signature-terms}}
    \label{fig:monotone-function}
\end{figure}

Next, suppose we can write $\gamma$ as the concatenation of two paths, $\gamma = \varphi \ast \phi$, with $\varphi\in\Pp_{a,t}^d$ and $\phi\in\Pp_{t,b}^d$ both satisfying Lemma~\ref{lem:sum-of-second-order-signature-terms} (for example both are monotone) - we refer to Figure~\ref{fig:area-via-integrals}. Dropping the subscript $a, b$ for brevity of notation, we may write
$$
\Ss(\varphi)^i\Ss(\varphi)^j = \Ss(\varphi)^{i, j} + \Ss(\varphi)^{j, i}
\qquad\text{and}\qquad
\Ss(\phi)^i\Ss(\phi)^j = \Ss(\phi)^{i, j} + \Ss(\phi)^{j, i}.
$$
A well-known result from Chen (see, for example, ~\cite[Section 1.3.3]{Chevyrev16}) establishes a connection between the signature of a concatenation of paths and the operation~$\tens$ of Section~\ref{sec:log-signature}, and states $\Ss(\varphi\ast \phi) = \Ss(\varphi) \tens \Ss(\phi)$.
Recall the non-commutative formal polynomial of $\Ss(\varphi)$:
\begin{equation}
    \Ss(\varphi) = 1 + S^1(\varphi)e_1 + \ldots + S^d(\varphi)e_d + S^{1,1}(\varphi)e_1e_1 + S^{1,2}(\varphi)e_1e_2 + \ldots
\end{equation}
and the similar representation for $\Ss(\phi)$. The coefficient of $e_i$ in the product $\Ss(\varphi) \tens \Ss(\phi)$ is seen to be $S^i(\varphi) + S^i(\phi)$, that is $\Ss(\gamma)^i = \Ss(\varphi \ast \phi)^i = \Ss(\varphi)^i + \Ss(\phi)^i$. Note that the geometric interpretation here is simply that the displacement in the $i\textsuperscript{th}$ coordinate path in the concatenation of~$\varphi$ and~$\phi$ is the sum of displacements in the path~$\varphi$ and~$\phi$.
We can also compute $\Ss(\varphi\ast \phi)^{i, j}$ in a similar fashion, obtaining
$$
\Ss(\varphi\ast \phi)^i = \Ss(\varphi)^i + \Ss(\phi)^i
\qquad\text{and}\qquad
\Ss(\varphi\ast \phi)^{i, j} = \Ss(\varphi)^{i}\Ss(\phi)^{j} + \Ss(\varphi)^{i,j} + \Ss(\phi)^{i,j},
$$
from which we have
    \begin{align*}
\Ss(\gamma)^{i,j} + \Ss(\gamma)^{j, i} &= \Ss(\varphi\ast \phi)^{i, j} + \Ss(\varphi\ast \phi)^{j, i} \\
&= \Ss(\varphi)^i \Ss(\phi)^j + \Ss(\varphi)^{i, j} + \Ss(\phi)^{i, j} + \Ss(\varphi)^j \Ss(\phi)^i + \Ss(\varphi)^{j, i} + \Ss(\phi)^{j, i} \\
&= \Ss(\varphi)^i \Ss(\phi)^j + \Ss(\varphi)^i \Ss(\varphi)^j + \Ss(\varphi)^j \Ss(\phi)^i + \Ss(\phi)^i \Ss(\phi)^j \\
&= (\Ss(\varphi)^i + \Ss(\phi)^i)(\Ss(\varphi)^j + \Ss(\phi)^j)
= \Ss(\varphi\ast \phi)^i \Ss(\varphi\ast \phi)^j
= \Ss(\gamma)^i \Ss(\gamma)^j.
    \end{align*}
Inductively, this proves the result for any curve composed of segments which are piecewise-monotone.
}

\hfill \ensuremath{\Box}


\bibliographystyle{abbrv}
\bibliography{biblio}

\end{document}